\documentclass[12pt]{article}

\usepackage[float]{preamble}
\usepackage{flafter} 
\usepackage{natbib}
\newcommand{\rEbase}{1.07}
\newcommand{\rElo}{0.45}
\newcommand{\rEhi}{2.52}

\newcommand{\rClo}{0.60}
\newcommand{\rChi}{8.22}

\newcommand{\rstarlo}{0.44}
\newcommand{\rstarhi}{2.21}
\newcommand{\rEcapcut}{6\%}
\newcommand{\pEbindspct}{80\%}
\newcommand{\rCrEratio}{3.2}
\newcommand{\rEdoubling}{2.1}
\newcommand{\rCalphadoubling}{2.5}
\newcommand{\rCdoubling}{11}

\newcommand{\rCbase}{3.45}
\newcommand{\gnubase}{0.92}
\newcommand{\gRbase}{0.86}
\newcommand{\ghbase}{1.93}

\newcommand{\sigmabase}{0.61}
\newcommand{\sigmahigh}{0.83}

\newcommand{\pbothpct}{49\%}
\newcommand{\pCpct}{85\%}
\newcommand{\pEpct}{56\%}

\newcommand{\ponlyEpct}{7\%}
\newcommand{\ponlyCpct}{37\%}
\newcommand{\pneitherpct}{7\%}
\newcommand{\rstarmedian}{0.99}
\newcommand{\gnuscalefree}{0.09}
\newcommand{\gnuimagenet}{0.92}
\newcommand{\gnupretraining}{0.99}
\newcommand{\gnuinference}{1.17}
\newcommand{\gnusd}{0.40}
\newcommand{\gRlow}{0.64}
\newcommand{\gRlowfactor}{1.90}

\newcommand{\sswbase}{0.46}
\newcommand{\sswlow}{0.31}
\newcommand{\sswhigh}{0.71}
\newcommand{\kdoublingpct}{60\%}
\newcommand{\kbase}{1.31}
\newcommand{\kbaseloglog}{1.29}
\newcommand{\kthin}{0.52}
\newcommand{\krct}{1.87}
\newcommand{\kgdpval}{1.34}
\newcommand{\sieshareGnu}{45\%}
\newcommand{\sieshareOthers}{24\%}

\newcommand{\sieGnuLowpct}{0\%}
\newcommand{\sieGnuHighpct}{85\%}
\newcommand{\sat}{\alpha^{*}}                    
\newcommand{\hbound}[1]{\overline{R}\!\left(#1\right)} 
\newcommand{\dd}{\,\mathrm{d}}

\usepackage{xcolor}

\newif\ifeditsfinal
\editsfinalfalse
\newcounter{edit}
\ifeditsfinal
  
  \newcommand{\deleted}[1]{}
  
  \newcommand{\moved}[1]{}

  \newenvironment{edit}{}{}
\else
  \newcommand{\editmark}{\stepcounter{edit}\textsuperscript{\textcolor{orange}{\sffamily\bfseries\tiny E\theedit}}}
  
  \newcommand{\deleted}[1]{\editmark\textcolor{red}{\sffamily\footnotesize[del: #1]}}
  
  \newcommand{\moved}[1]{\editmark\textcolor{violet}{\sffamily\footnotesize[#1]}}

\fi

\newtheorem{proposition}{Proposition}
\newtheorem{lemma}{Lemma}
\newtheorem{corollary}{Corollary}
\theoremstyle{definition}

\newtheorem{definition}{Definition}
\theoremstyle{remark}

\title{Efficiency vs Capability in the Intelligence Explosion}
\author{Karthik Tadepalli\footnote{GovAI, \href{mailto:karthik.tadepalli@governance.ai}{karthik.tadepalli@governance.ai}. This paper benefited greatly from the feedback of Sam Manning, Cheryl Wu, Alan Chan, Markus Anderljung, Tom Davidson, Phil Trammell, Tom Cunningham, Joseph Levine, Kris Gulati, and audiences at OpenAI and Forethought. Any remaining errors are my own. Replication code for all figures and tables is available \href{https://github.com/karthiktadepalli1/efficiency-capability-sie}{here}.}}
\date{\today}

\begin{document}
\maketitle

\onehalfspacing

\begin{abstract}
\noindent AI progress can feed back into AI research in two ways. Capability improvements allow AIs to perform a larger share of AI research tasks, while efficiency improvements raise AI research output per unit of compute on tasks already automated. In a task automation model of AI research, I show that a \textit{software-driven intelligence explosion}---a self-sustaining acceleration in AI progress with fixed compute and human labor---requires both margins to improve sufficiently fast, and neither margin can substitute for the other. To see whether these conditions hold, I calibrate both margins with data on recent AI progress. I estimate that capability progress is likely fast enough to support an SIE, while efficiency progress is on the knife-edge. In Monte Carlo simulations, an SIE occurs in \pbothpct{} of draws, indicating that an SIE is plausible but uncertain based on historical trends in AI progress. The single largest source of uncertainty is the rate of efficiency progress, whose uncertainty is responsible for \sieshareGnu{} of the uncertainty about whether an SIE occurs, more than all other parameters combined. Thus, better data on efficiency progress in frontier models is essential for understanding the likelihood of an SIE.
\end{abstract}

\section{Introduction}\label{sec:intro}

Frontier AI systems are increasingly being used to accelerate R\&D into future AI systems. This creates a positive feedback loop where the development of increasingly capable AI systems can aid in the development of even more advanced successor systems. Some version of this \textit{recursive self-improvement} loop already operates today, with AI companies reporting internal research speedups from using AI  \citep{favaro2026ai,openai2026acceleration}. If this loop becomes strong enough to sustain itself---accelerating AI progress even with no growth in compute or human researchers---the result would be a \textit{software-driven intelligence explosion} (SIE), potentially leading to runaway AI progress. Understanding the conditions for an SIE is one of the most important questions in AI today. 

At the same time, there is no commonly agreed-upon yardstick of AI research automation, and no clear understanding of what data we would need to measure the recursive self-improvement loop. Frontier developers each report different statistics on research acceleration, and it is unclear whether these statistics are the most important ones for forecasting an SIE. Thus, we need a conceptual framework not only to tell us whether to expect an SIE, but also to understand what data will actually resolve uncertainties around an SIE. 

To establish the conditions for an SIE and key sources of uncertainty, I formalize a model of AI research in which AI progress feeds back into itself through two margins. \textit{Capability} is the share of AI research tasks that AI can perform, while \textit{efficiency} is AI research output per unit of compute on automated tasks. I use this model to establish the conditions for an SIE, and then calibrate the model with data on AI progress to see whether these conditions hold in recent AI progress, and which parameters are the most important for resolving uncertainty about an SIE.

The benchmark model of idea production \citep{jones1995rd} summarizes research with a scalar stock of efficiency units. The pace of technological progress is entirely a function of this stock, with no distinction between types of researchers. In this representation, sufficiently many low-productivity researchers substitute for any high-productivity researcher. The same logic makes sufficiently many copies of a less capable AI a substitute for a more capable model (or human). For example, Alec Radford is a renowned AI researcher who led the development of OpenAI's original GPT models. But the benchmark model implies there is some $N$ such that $N$ freshman computer-science majors do better AI research than Alec Radford. This is an implausible assumption in the context of AI research. 

I replace scalar research progress with the task-based model of \citet{aghion2019artificial}, applied to AI research itself, with compute in place of capital. AI progress feeds back into AI research through two channels. Capability improvements expand the set of AI research tasks that AI can successfully perform; efficiency improvements increase the amount of AI research labor supplied by a given compute stock, so they are analogous to population growth for automated researchers. The compute that matters for research automation is \emph{inference} compute, the compute used to run trained AI models, as distinct from the \emph{training} compute used to create them. I show that the speed of AI progress is determined by the slower of these two margins, so an SIE requires both margins to be sufficiently fast. Efficiency improvements alone eventually run into Baumol effects, where non-automated tasks bottleneck AI progress. Capability improvements alone fail because automating a larger share of tasks only spreads the fixed compute stock thinly across all tasks, without expanding the population of automated researchers.

In the second half of the paper, I quantify efficiency progress and capability progress using recent data on AI progress, to test whether the conditions for an SIE hold. Each leg of the SIE condition can be expressed as a ratio of an observable growth rate to the growth rate of research inputs, in the manner of the returns-to-software parameter of \citet{eth2025software}. The efficiency leg is the rate of algorithmic efficiency progress on a fixed set of tasks, which I take from the literature. My main empirical contribution is to quantify capability growth in a way that connects to the conditions for an SIE. To do so, I use the \emph{time horizon} measured by METR, an AI evaluation nonprofit: the length of task, in the time it takes a skilled human, that an AI model completes with 50\% probability. On a suite of software and research tasks with measured human completion times, the time horizon of frontier models has doubled every four to seven months \citep{metr2026th11}. Combining this horizon trend with the distribution of human completion times across tasks, and the assumption that this distribution has a Pareto tail, gives the growth rate of capabilities.

In my baseline calibration, point estimates imply both the efficiency and capability conditions for an SIE are met: capability returns are $\rCbase$ (95\% interval $\rClo$--$\rChi$), and efficiency returns are $\rEbase$ ($\rElo$--$\rEhi$). Thus, efficiency rather than capability is the binding constraint on an SIE, but the 95\% intervals include both conditions failing. In a Monte Carlo simulation over the uncertainty in each input, I find that an SIE occurs in \pbothpct{} of draws. Thus, an SIE is plausible. The more robust finding is that efficiency, rather than capability, is the likely bottleneck: the efficiency returns are the lower of the two in \pEbindspct{} of draws.

The main takeaway of the quantitative exercise is therefore not a precise number for the likelihood of an SIE, but a magnifying glass on the most important uncertainty behind it: the rate of efficiency progress. Learning only the rate of efficiency progress would resolve \sieshareGnu{} of the uncertainty about whether an SIE occurs, more than all other inputs combined, which together would resolve \sieshareOthers{}. Across draws, the share that are an SIE rises from \sieGnuLowpct{} among those with the slowest fifth of efficiency progress to \sieGnuHighpct{} among those with the fastest. It is also the input on which existing estimates disagree most, by an order of magnitude across sources. Better measurement of efficiency progress, and of how it responds to research effort, is therefore the most valuable step toward knowing whether an SIE is plausible.

The rest of the paper is organized as follows. Section~\ref{sec:literature} relates the paper to the literature. Section~\ref{sec:env} lays out the model. Section~\ref{sec:sie} derives the two-legged condition for an SIE. Section~\ref{sec:quant} maps the model to data and estimates the two margins of progress. Section~\ref{sec:conclusion} concludes and discusses policy implications.

\section{Related Literature}\label{sec:literature}

This paper is part of the still-too-small literature that models AI progress using endogenous growth theory. Canonically, these papers \citep{eth2025software,davidsonhoulden2025quick,aifutures2025takeoff} have used a benchmark endogenous growth model in the style of \citet{jones1995rd}, which measures research in scalar efficiency units and so elides the capability-efficiency distinction that I focus on. Some recent papers in this literature have taken a different path. \citet{davidson2026networks} focus on networked spillovers and feedback loops in the macroeconomic setting, with AI as one sector in an economywide feedback loop; they do use a task automation model to understand the importance of automation speed (which is analogous to capability growth rates in my model). \citet{cunningham2026apples} model AI research as an apple-picking process, where AI can pick apples up to a certain ``height'', corresponding to task automation in my model. However, they model AI labor as abundant, ignoring the efficiency dimension that I focus on. \citet{cunningham2026rsi} apply a similar framework to model AI research automation; they do highlight the capability-efficiency distinction. However, their framework focuses on measuring reduced-form elasticities, whereas I derive structural elasticities from a task-based model. This requires more assumptions, but is better able to use historical data. 

Methodologically, my model is a task automation model, with similarities to others deployed in the literature, including those used to study the effect of AI on economic growth \citep{aghion2019artificial, jonestonetti2026weaklinks}. The mechanics of my model are almost identical to those; the feedback loop over effective compute echoes the positive feedback from capital in a standard growth model, and many capital-accumulation mechanisms carry over to the setting of AI research.  My main theoretical contribution is to show that the task automation structure implies a two-legged condition for an SIE. Secondarily, while previous papers start from the premise of full automation of AI research to derive the conditions for an SIE \citep{eth2025software,davidsonhoulden2025quick}, I derive conditions based on the \emph{speed} of AI research automation, which is observable today. 

Empirically, I draw on the literature that estimates AI progress: fixed-task estimates of algorithmic efficiency \citep{ho2024algorithmic,erdil2022cv,gundlach2025origin}, the time-horizon measurements of \citet{kwa2025horizons} and \citet{metr2026th11} and their decomposition into compute and software \citep{whitfill2025horizon}, and estimates of the substitutability of compute for research labor \citep{whitfillwu2025compute}. In calibrating the efficiency leg of AI progress, I largely rely on this literature. My main empirical contribution is to connect the growth of AI capabilities to data on the time horizons of tasks that AI can complete, and using this calibration to assess the likelihood of an SIE.

\section{Model}\label{sec:env}

AI development requires a continuum of research tasks $i \in [0,1]$, to be performed by a fixed stock $L>0$ of human researchers or a fixed stock $Z>0$ of physical compute (on which copies of the frontier AI model run).\footnote{I hold physical compute fixed because on the short timescales for which people worry about an SIE, compute is unlikely to be part of the feedback loop.}

A key feature of my model is that ``compute'' denotes \emph{inference compute used to run automated researchers}. This is an important deviation from standard ways of thinking about the role of compute in AI progress. Compute is most commonly identified with training compute, because scaling laws suggest that scaling training compute increases AI capabilities. It is also sometimes identified with experimental compute, which could be a strong complement to research labor, and thus could bottleneck AI progress \citep{whitfillwu2025compute}. While both of these are reasonable uses of the term ``compute'', neither is applicable to this model. This is an important restriction: it means that my model does not describe the conditions under which training compute bottlenecks or experimental compute bottlenecks could prevent an SIE. My model focuses only on the conditions for when AI cognitive labor can substitute for human cognitive labor. A full accounting of all the roles that compute plays in AI research, and an endogenous allocation between all of these roles, is left to future research.

\paragraph{AI progress.} The rate of change $\dot{A}_t$ of frontier AI technology is driven by the level of current technology $A_t$ and research output $R_t$, as in the benchmark semi-endogenous model.

\begin{equation}\label{eq:innovation}
\dot{A}_t \;=\; \delta\, R_t^{\lambda}\, A_t^{1-\beta}, \qquad \delta > 0,\ \beta > 0,\ 0 < \lambda \le 1
\end{equation}

Research output enters \eqref{eq:innovation} with elasticity $\lambda$, the stepping-on-toes parameter: when $\lambda<1$, parallel research effort duplicates itself, so doubling $R_t$ less than doubles the rate of progress. Technology enters with elasticity $1-\beta$: doubling $A_t$ less than doubles the rate of progress, representing ideas getting harder to find.

\paragraph{Research production.} AI research requires performing tasks indexed by $i \in [0, 1]$. Research task $i$ can either be human-only, or automatable (can be performed by AI). Task-specific output is given by

\begin{equation}\label{eq:task}
r_t(i) \;=\; \begin{cases} \ell_t(i) + \nu(A_t)\, z_t(i), & i \le \alpha(A_t), \\[2pt] \ell_t(i), & i > \alpha(A_t) \end{cases}
\end{equation}

Here, $\nu(A_t)$ is efficiency, the task output that can be produced with one unit of compute on an automated task, and $\alpha(A_t)$ is capability, the share of tasks that can be performed by AI at all; both depend on the level of AI technology, as specified below.

Tasks are complements in aggregate research output, being combined with a constant elasticity of substitution $0<\sigma<1$:
\begin{equation}\label{eq:ces}
R_t \;=\; \left( \int_0^1 r_t(i)^{\frac{\sigma-1}{\sigma}} \dd i \right)^{\frac{\sigma}{\sigma-1}}
\end{equation}

\paragraph{Capability and efficiency improvements.} Efficiency $\nu$ and capability $\alpha$ are determined by the level of AI technology:

\begin{equation}\label{eq:breadth}
\alpha(A_t) = 1 - A_t^{-\mu}, \qquad \mu \ge 0,\ A_t \ge A_0 \ge 1,
\end{equation}

\begin{equation}\label{eq:depth}
\nu(A_t) = c\, A_t^{\eta}, \qquad c > 0,\ \eta \ge 0
\end{equation}

Both of these equations parameterize capability and efficiency in terms of elasticities. $\mu$ is the capability elasticity and $\eta$ is the efficiency elasticity: a 1\% increase in the technology level raises AI research input per unit of compute by $\eta$\% and shrinks the human-only task share $1-\alpha$ by $\mu$\%.

Of course, parameterizing capability and efficiency this way assumes that their elasticities are constant, no matter the level of technology. This is a particularly important assumption when it comes to capability: Equation~\eqref{eq:breadth} implies that $\alpha < 1$, and thus AI research is \emph{never} fully automated. \citet{eth2025software} start from the premise of fully-automated AI research, and then derive the conditions for an SIE, so this may at first seem to preclude the possibility of an SIE. As I show below, this is not true; rather, it casts the condition for an SIE in terms of \textit{how fast capabilities grow} instead of \textit{whether research has been fully automated}, and the former is much easier to measure, which makes empirical analysis of an SIE more tractable.

Define \emph{effective compute} as the AI research input supplied by physical compute at the current level of efficiency,
\[
X_t \;\equiv\; \nu(A_t)\, Z ,
\]
with units chosen so that $X_t$ is measured in researcher-equivalents. Its physical component $Z$ is a fixed stock, and its software component $\nu(A_t)$ is produced by the loop itself.

Finally, the resource constraint is that researchers and physical compute applied across tasks exhaust the total available researchers and compute. 

\begin{equation}\label{eq:labor}
\int_0^1 \ell_t(i) \dd i \;=\; L,
\end{equation}
\begin{equation}\label{eq:compute}
\int_0^{1} z_t(i) \dd i \;=\; Z
\end{equation}
Table~\ref{tab:environment} collects the environment.

\begin{table}[t]
\centering
\small
\renewcommand{\arraystretch}{1.35}
\begin{tabular}{@{}p{0.39\textwidth}p{0.55\textwidth}@{}}
\toprule
Block & Statement \\
\midrule
Technology: law of motion & $\dot{A}_t = \delta R_t^{\lambda} A_t^{1-\beta}$ \\
Technology: capability & $1 - \alpha(A_t) = A_t^{-\mu}$ \\
Technology: efficiency & $\nu(A_t) = c A_t^{\eta}$ \\
Research tasks & $r_t(i) = \ell_t(i) + \mathbf{1}\{i \le \alpha(A_t)\} \cdot \nu(A_t)z_t(i)$ \\
Research aggregation & $R_t = \big(\int_0^1 r_t(i)^{\frac{\sigma-1}{\sigma}} \dd i\big)^{\frac{\sigma}{\sigma-1}}$ \\
Resources: researchers & $\int_0^1 \ell_t(i)\dd i = L$ \\
Resources: physical compute & $\int_0^{1} z_t(i)\dd i = Z$ \\
\bottomrule
\end{tabular}
\caption{The economic environment.}
\label{tab:environment}
\end{table}

To close the model, we need an allocation rule for labor and compute. I assume that the allocation is chosen to maximize aggregate research output $R$, consistent with labor and compute being allocated internally to a profit-maximizing AI developer. Define \emph{research capacity} as this maximized $R^*$.

\begin{proposition}[Task Allocation]\label{prop:allocation}
Define $\sat \equiv \frac{X}{L+X}$. Research capacity is
\begin{equation}\label{eq:Rfunction}
R^*(\alpha; L, X) \;=\;
\begin{cases}
\left[\, \alpha^{\frac{1}{\sigma}} X^{\frac{\sigma-1}{\sigma}} + (1-\alpha)^{\frac{1}{\sigma}} L^{\frac{\sigma-1}{\sigma}} \,\right]^{\frac{\sigma}{\sigma-1}}, & \alpha \le \sat,\\[4pt]
L + X, & \alpha \ge \sat.
\end{cases}
\end{equation}
\end{proposition}
\noindent\emph{Proof.} See Appendix~\ref{app:proof-allocation}.

Below the threshold, research capacity is a CES aggregate of the two factor stocks. The factor weights $(1-\alpha)^{1/\sigma}$ on labor $L$ and $\alpha^{1/\sigma}$ on effective compute determine how important each factor is in generating research output. Capability improvements (increases in $\alpha$) thus shift the weight from labor $L$ to effective compute $X$. Assuming $X > L$, this means that increases in $\alpha$ increase research capacity $R^*$, even if both $X$ and $L$ are fixed. The labor-augmented term $L(1-\alpha)^{-1/(1-\sigma)}$ is the most research that can be produced at capability $\alpha$, even if effective compute becomes abundant (Lemma~\ref{lem:bounds} in Appendix~\ref{app:proof-bounds}), and it is essential for the capability leg of the explosion condition below.

The threshold $\sat = X/(L+X)$ is the point of full factor utilization. Past it, compute has been spread so thin that per-task AI input falls below what concentrated researchers supply elsewhere; researchers re-enter automated tasks until inputs equalize, and research capacity is just the total factor supply, $L+X$.

To visualize the relationship between efficiency and capability in producing research effort, Figure~\ref{fig:output-map} plots (normalized) research effort $R^*/L$ as a function of (normalized) efficiency $X/L$ and capability $\alpha$, assuming an elasticity $\sigma=\sigmabase$ that will eventually be used in the calibration of Section~\ref{subsec:capability-leg}. This demonstrates what will later become explicit---that research effort is only high when both efficiency and capability are high. 

\begin{figure}[tbp]
\centering
\includegraphics[width=0.85\textwidth]{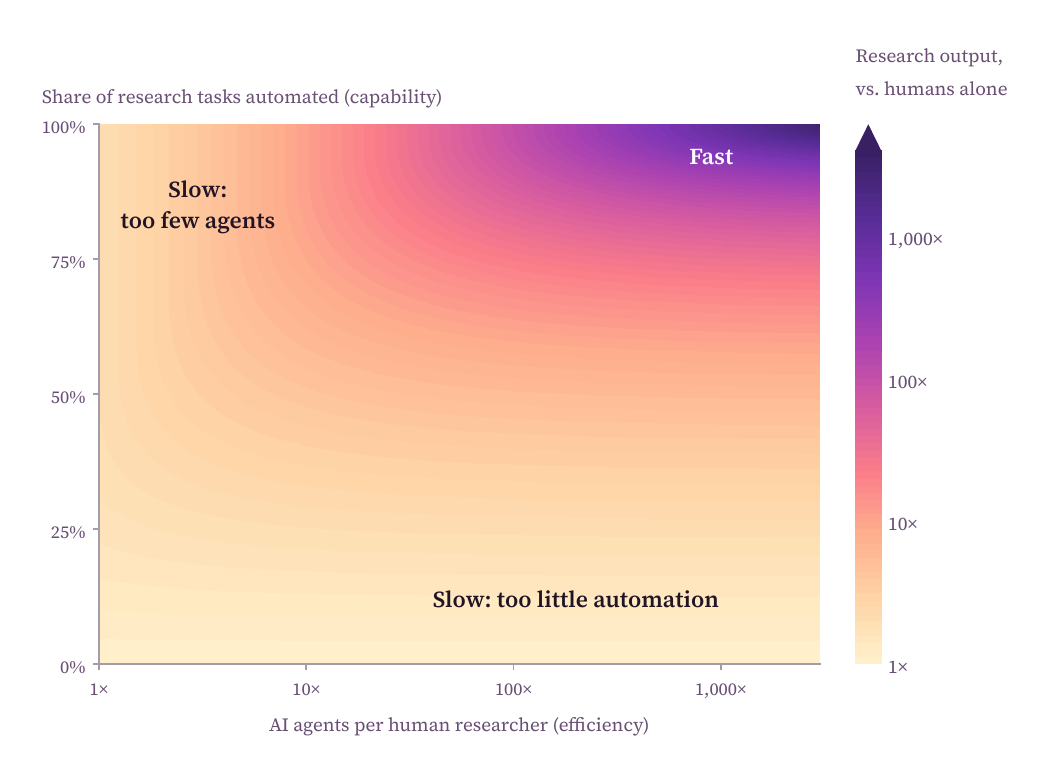}
\caption{Normalized research capacity $R^*/L$ as a function of normalized efficiency $X/L$ and capability $\alpha$, at $\sigma=\sigmabase$.}
\label{fig:output-map}
\end{figure}

\section{Explosion Conditions}\label{sec:sie}

\begin{definition}[Software Intelligence Explosion]\label{def:sie}
Holding $L$ and $Z$ fixed, the model exhibits a software intelligence explosion if AI progress eventually raises the growth rate of the technology level, $g_A = \dot A_t/A_t$, at a sustained proportional rate:
\begin{equation}\label{eq:sie-definition}
\liminf_{A\to\infty}\frac{\dd\ln g_A}{\dd\ln A}>0
\end{equation}
Locally, acceleration occurs when
\[
\frac{\dd\ln g_A}{\dd\ln A}>0
\quad\Longleftrightarrow\quad
\frac{\dd\ln\dot A}{\dd\ln A}>1
\]
The limit condition in \eqref{eq:sie-definition} distinguishes a self-sustaining acceleration from a temporary rise in $g_A$ or convergence to a constant growth rate.
\end{definition}

Another reasonable definition of an SIE would be as a technological singularity, where $A_t$ becomes infinite in finite time. In this model, there is no distinction; the results below hold for both definitions.

\begin{proposition}[Asymptotic Growth]\label{prop:asymptotic}
Define the efficiency returns, the capability returns and the binding returns as
\[
r_E\equiv\frac{\lambda\eta}{\beta},\qquad r_C\equiv\frac{\lambda\mu}{(1-\sigma)\beta},\qquad r^{*}\equiv\min\{r_E,r_C\}.
\]
Holding $L$ and $Z$ fixed,
\begin{equation}\label{eq:asymptotic}
\lim_{A\to\infty}\frac{\dd\ln g_A}{\dd\ln A}\;=\;\beta\,(r^{*}-1),
\qquad\text{equivalently}\qquad
\lim_{A\to\infty}\frac{\dd\ln g_A}{\dd\ln A^{\beta}}\;=\;r^{*}-1 .
\end{equation}
\end{proposition}
\noindent\emph{Proof.} See Appendix~\ref{app:proof-asymptotic}.

The elasticity of research capacity with respect to technology converges to the smaller of two numbers: $\eta$, the rate at which effective compute grows, and $\mu/(1-\sigma)$, the rate at which the shrinking human-only task share relaxes its constraint on research. Duplication scales that elasticity by $\lambda$, and idea difficulty subtracts $\beta$. Integrating \eqref{eq:innovation} gives $A_t^{\beta}=A_0^{\beta}+\beta\delta\int_0^t R_\tau^{\lambda}\dd\tau$, so $A^{\beta}$ grows with cumulative research net of duplication, so the second form of \eqref{eq:asymptotic} reads: each percent of cumulative research raises the growth rate by $r^{*}-1$ percent. When $r^{*}<1$ the growth rate falls at that rate and the technology level grows sub-exponentially.

The obstacle that an SIE has to overcome is diminishing returns to research effort, captured by $\beta$: ideas get harder to find \citep{bloom2020ideas}. Each idea can only be found once, and AI research has used up some of its easiest gains, such as removing duplicated text from web-scraped training data \citep{lee2022dedup} and correcting the balance between model size and training data for a given compute budget \citep{hoffmann2022chinchilla}. Each further improvement in AI therefore requires more research effort than the last.

The two ratios are the objects that matter for the empirical analysis in Section~\ref{sec:quant}: the technology level has no natural units, so the elasticities $\eta$, $\mu$ and $\beta$ are identified from data only up to a common scale and $\lambda$ only jointly with them, whereas $r_E$ and $r_C$ are identified directly.

\begin{corollary}[SIE Conditions]\label{cor:sie}
The model exhibits a software intelligence explosion if and only if $r^{*}>1$, that is,
\begin{equation}\label{eq:jointgeometric}
\boxed{\;
\underbrace{\lambda\eta>\beta}_{\text{efficiency leg}}
\quad\text{and}\quad
\underbrace{\lambda\mu>(1-\sigma)\beta}_{\text{capability leg}}
\; }
\end{equation}
The same conditions are necessary and sufficient for $A_t$ to diverge in finite time. If $r^{*}=1$, the growth rate converges to a positive constant.
\end{corollary}
\noindent\emph{Proof.} See Appendix~\ref{app:proof-sie}.

Each leg is a rate of return: if cumulative research effort doubles, how many times does the margin double? For efficiency, the margin is effective compute $X$, the supply of automated research labor. For capability, it is the human-bottleneck bound $\hbound{\alpha}=L(1-\alpha)^{-1/(1-\sigma)}$ of Lemma~\ref{lem:bounds}: the most research capacity that the economy could produce at capability $\alpha$ if effective compute were unlimited. Capability progress raises this ceiling but supplies no automated labor; research capacity approaches the ceiling only as effective compute grows. When the return exceeds one, each doubling of research effort yields more than one doubling of the margin, which more than pays for the next doubling of research effort, so progress accelerates. When it falls below one, each round of progress buys less than the next round costs, and growth slows.

Each inequality rules out a different failure. When either inequality fails strictly and is the tighter constraint, that margin eventually pulls the growth rate down.

\emph{Efficiency alone.} Suppose efficiency improves fast enough to overcome idea difficulty, but capability does not: $\lambda\eta>\beta$ but $\lambda\mu\le(1-\sigma)\beta$. Productivity rises rapidly on automated tasks, but the complementary human-only task share contracts too slowly. Those tasks become an increasingly important share of the effective cost of research and eventually govern aggregate progress. This is the familiar Baumol effect, and the reason why \citet{aghion2019artificial} emphasize full automation as the condition for explosive economic growth: rapid progress in some tasks cannot sustain aggregate growth when essential laggard tasks remain.

Formally, research capacity never exceeds $L(1-\alpha)^{-1/(1-\sigma)}$ (Lemma~\ref{lem:bounds} in Appendix~\ref{app:proof-bounds}), so for any efficiency path,
\begin{equation}\label{eq:capability-failure-bound}
g_A
\le
\delta L^{\lambda} A^{\lambda\mu/(1-\sigma)-\beta}
\end{equation}
Since $\lambda\mu/(1-\sigma)-\beta \le 0$ in this scenario, $g_A$ is bounded by $\delta L^{\lambda}$ at equality and falls to zero when the inequality is strict. Growth slows down as humans become the limiting factor.

\emph{Capability alone.} The failure of capability improvements without efficiency improvements is more interesting, and is instructive about the difference between AI research and the economic growth context that this model is derived from. Capability improvements enter research capacity only through the weights in \eqref{eq:Rfunction}: they spread the fixed rival compute stock over a larger set of tasks, the dilution $\alpha^{-1/(1-\sigma)}X$ in the factor-augmenting form above, but add no effective compute. The role of efficiency improvements is therefore analogous to the role of population growth in a semi-endogenous growth model: they are population growth for AI labor, expanding the number of automated researchers available to do AI research. Without them, the population of automated researchers is fixed, and a fixed population cannot generate exponential growth in a semi-endogenous model. This is why efficiency improvements are essential for sustained accelerations in AI progress.

Formally, research capacity never exceeds total factor supply, $R^*\le L+X$ (Lemma~\ref{lem:bounds}), and effective compute is $X=cZA^{\eta}$, so for any capability path
\begin{equation}\label{eq:efficiency-failure-bound}
g_A
\le
\delta\,(L+cZA^{\eta})^{\lambda}\,A^{-\beta},
\end{equation}
which falls to zero when $\lambda\eta<\beta$ and is bounded by $\delta(L+cZ)^{\lambda}$ when $\lambda\eta=\beta$.

\section{Calibration}\label{sec:quant}

\subsection{Empirical Mapping}\label{subsec:mapping}

To map the theory to data, I first clarify the empirical counterparts to the model objects. In particular, ``capability'' and ``efficiency'' have subtle differences from the way these terms are ordinarily used. \citet{eth2025software} classify both acquiring new tasks and becoming more competent at tasks a system can already attempt as capability improvements, while achieving a given capability level with less compute is an efficiency improvement. Those definitions are intuitive, and useful for a broad accounting of software progress. But for the purpose of SIE accounting in this model, a different boundary is necessary.

\paragraph{Efficiency.} Efficiency is useful task output per unit of compute on automated tasks. The key feature of this definition is that the efficiency that matters for an SIE is \emph{inference} efficiency, since that determines the quantity of AI labor available for AI research. Thus, the empirical counterpart I use is primarily the rate at which the inference compute needed to reach a fixed benchmark score has fallen \citep{gundlach2026price}, triangulated with estimates of efficiency progress based on pretraining (Section~\ref{subsec:efficiency-leg}) even though these are not the relevant margin of efficiency progress.

The second important and unintuitive element of this definition is that improved quality at performing the same task is an \emph{efficiency} improvement, rather than a capability improvement. If a model's code quality improves for the same compute cost, this may seem like a capability improvement, but for SIE accounting it raises efficiency. Within a task, I assume that quality and cost are interchangeable, so that an AI that can do a task outstandingly well is isomorphic to an AI that can do that task adequately at low cost. Thus, the critique that motivated this model---that the standard model treats quality and quantity as interchangeable---still applies. However, it is a more palatable assumption at the task level than at the level of AI research as a whole. On the whole, this exercise largely tracks previous exercises that aim to calculate the returns to software progress as they relate to an SIE (see Table~\ref{tab:published-r}), and is not conceptually novel. 

\paragraph{Capability.} In this model, capability is the automated task share $\alpha$. Unlike efficiency, the automated share is not observed: there is no comprehensive enumeration of AI research tasks, let alone a record of which of them are automated. \citet{anthropic2026pace} take a step in this direction: they catalogue Anthropic's model R\&D work into 378 task categories, weight each by the staff time spent on it, and rate how far each is automated. They publish only aggregate shares, however, not the task list or the per-task ratings, so their task set cannot be used here. So how can we measure the automated share of tasks? 

Mapping the unobservable task share growth to observable quantities is my main empirical contribution. To do this, we have to augment the model of Section~\ref{sec:env} with four assumptions.

\begin{enumerate}[label=(\alph*)]
\item \emph{Task length.} Index each research task by a scalar difficulty $q>0$, measured as the time a human expert needs to complete it (``task length''). Task length is distributed $q\sim F$.
\item \emph{Capability threshold.} A model at technology level $A$ has some threshold difficulty $h(A)$ such that it can perform every task with $q\le h(A)$ and no task with $q>h(A)$. Define $h(A)$ as the model's \emph{time horizon}.\footnote{\citet{ord2025halflife}'s model of success rates over a time horizon provides a microfoundation for this step: if AIs have a constant probability of failure per unit of human task-time, then the AI success probability falls with human task length, and a single length $h(A)$ separates the tasks a model completes at a given reliability from those it does not.}
\item \emph{Task measure.} The continuum in \eqref{eq:ces} weights tasks by their labor requirement: a task that takes $T$ hours occupies $T$ units of researcher input. $F$ is therefore the distribution of task length across work-hours, $F(T)$ being the fraction of all task-hours in tasks with $q\le T$, and the automated share is $\alpha(A)=F(h(A))$.
\item \emph{Pareto tail.} Above some length $h_0$, $F$ has a Pareto tail with index $k>0$: $1-F(T)=\bar F_0\,(T/h_0)^{-k}$ for $T\ge h_0$, where $\bar F_0\equiv1-F(h_0)\in(0,1]$.
\end{enumerate}

\begin{proposition}[Capability from Time Horizons]\label{prop:horizon}
Under (a)--(d), for every technology level with $h(A)\ge h_0$,
\begin{equation}\label{eq:alpha-horizon}
1-\alpha(A)=\bar F_0\left(\frac{h(A)}{h_0}\right)^{-k},
\qquad
\mu\;=\;k\,\frac{\dd\ln h}{\dd\ln A},
\end{equation}
and along any path of AI progress the human-only task share shrinks at rate
\begin{equation}\label{eq:galpha}
g_\alpha\;\equiv\;-\frac{\dd\ln(1-\alpha)}{\dd t}\;=\;k\,g_h,
\qquad g_h\equiv\frac{\dot h}{h}.
\end{equation}
\end{proposition}
\noindent\emph{Proof.} See Appendix~\ref{app:proof-horizon}.

The intuition is that when we learn how the time horizon of AI models evolves, we learn which task lengths it can do with some threshold probability. If we further know how research work is spread across task lengths, we can convert a longer horizon into a larger automated share. Under the Pareto assumption, each doubling of the horizon automates the same fraction of remaining work, whatever the current horizon. With $k=\kbase$, every doubling of the horizon automates about $1 - 2^{-k} =$ \kdoublingpct{} of the remaining human-only work. The human-only share therefore shrinks at a constant proportional rate, $k$ times the growth rate of the horizon. Thus, Proposition~\ref{prop:horizon} replaces the unobservable quantity $\alpha(A)$ with two quantities $k$ and $g_h$ that can be estimated from the data. This solves our problem of estimating the growth rate of capabilities.

One final adjustment is needed before we can use the horizon trend to estimate capability progress. The observed time horizon growth reflects both software progress and the growth of training compute, but when studying an SIE we want to know the capability growth that came purely from software progress. We can attribute some share $s$ of time horizon growth to software progress, so that the software-attributable horizon growth $s\,g_h$ is what enters the capability leg.

\subsection{Identification}\label{subsec:identification}

The parameters $\eta$, $\mu$, $\beta$, $\lambda$ are defined relative to the technology level $A$, which has no natural units. Only the combinations $\lambda\eta/\beta$ and $\lambda\mu/\beta$ are identified from data, and Corollary~\ref{cor:sie} states the SIE conditions in exactly those combinations, $r_E>1$ and $r_C>1$. In the data we can observe three growth rates that map into the model:

\begin{enumerate}[nosep]
\item the efficiency growth rate $g_\nu\equiv\dot\nu/\nu=\eta g_A$
\item the rate $g_\alpha=\mu g_A$ at which the human-only task share shrinks
\item the growth rate of research capacity $g_R\equiv\dot R/R$.
\end{enumerate}

Then the SIE condition---expressed in terms of these growth rates and the observable parameters of the model---is:

\begin{proposition}[Returns Ratios]\label{prop:returns}
Suppose that over a historical window research capacity and technology grow at constant rates $g_R>0$ and $g_A$. Then $\beta g_A=\lambda g_R$, the returns of Proposition~\ref{prop:asymptotic} are ratios of observable growth rates, and the SIE conditions of Corollary~\ref{cor:sie} read
\begin{equation}\label{eq:returns}
\boxed{\;
r_E\;=\;\frac{g_\nu}{g_R}\;>\;1
\quad\text{and}\quad
r_C\;=\;\frac{g_\alpha}{(1-\sigma)\,g_R}\;>\;1
\;}
\end{equation}
\end{proposition}
\noindent\emph{Proof.} See Appendix~\ref{app:proof-returns}.

Both legs are \emph{returns} parameters. The efficiency returns $r_E$ count the doublings of efficiency per doubling of cumulative research input along a balanced path, which is the returns-to-software parameter $r$ of \citet{eth2025software} and \citet{davidsonhoulden2025quick} applied to inference efficiency. The capability returns $r_C$ count the doublings of the human-bottleneck bound $\hbound{\alpha}=L(1-\alpha)^{-1/(1-\sigma)}$ (Lemma~\ref{lem:bounds}) per doubling of cumulative research input.

Proposition~\ref{prop:horizon} further gives $g_\alpha=k\,g_h$; with compute held fixed the relevant horizon growth rate is $s\,g_h$, so the capability leg reads
\begin{equation}\label{eq:rC-horizon}
r_C=\frac{k\,s\,g_h}{(1-\sigma)\,g_R} .
\end{equation}

The calibration uses six inputs, listed with their baselines and ranges in Table~\ref{tab:inputs}.

\begin{table}[t]
\centering
\small
\begin{tabular}{@{}llccl@{}}
\toprule
Symbol & Input & Baseline & Range & Distribution \\
\midrule
$g_\nu$ & Fixed-task efficiency growth & 0.92 & 0.42--2.02 & log-normal \\
$g_R$ & Research labor growth & 0.86 & 0.64--1.15 & log-uniform \\
$g_h$ & Time-horizon growth (50\%) & 1.93 & 1.57--2.36 & log-normal \\
$s$ & Software share of horizon growth & 0.46 & 0.31--0.71 & uniform \\
$k$ & Task-length tail index & 1.31 & 0.52--1.87 & log-uniform \\
$\sigma$ & Task substitution elasticity & 0.61 & 0.00--0.83 & uniform \\
\bottomrule
\end{tabular}

\caption{Calibration inputs. Rates are log points per year. Sources in Table~\ref{tab:sources}.}
\label{tab:inputs}
\end{table}

\subsection{Efficiency Leg}\label{subsec:efficiency-leg}

\paragraph{Efficiency progress $g_\nu$.} Unlike the returns-to-software estimates of \citet{eth2025software} and \citet{davidsonhoulden2025quick}, which use pretraining efficiency, the numerator of $r_E=g_\nu/g_R$ is primarily inference efficiency. I pool four estimates of algorithmic efficiency progress at a fixed task and performance level (Appendix~\ref{app:efficiency-detail} gives each source). The primary estimate measures the right object: the decline in the inference compute needed to reach a fixed score on GPQA Diamond, a benchmark of graduate-level science questions, net of hardware improvements, which \citet{gundlach2026price} infer from the prices of open-weight models. It is $\gnuinference$ per year. Because it is identified from a small set of open models over a short period, I triangulate it with three training-side estimates and pool the four; Appendix~\ref{app:efficiency-detail} describes the sources, the weights and the pooling. The baseline is $g_\nu=\gnubase$.

\paragraph{Research inputs $g_R$.} Staff at frontier laboratories has grown rapidly since the release of ChatGPT in late 2022. \citet{howhitfill2025experiments} report a growth rate of 0.85 log points per year in OpenAI's staff count, approximately $2.3\times$ per year. \citet{davidson2026networks} report growth in AI-laboratory headcount of roughly $140\%$ per year since ChatGPT (i.e.\ $\ln 2.4\approx0.88$ log points per year), from Epoch AI data. I set a baseline of $g_R=\gRbase$ log points per year. Note also that since $g_R$ is the denominator of both returns ratios, it appears in both SIE conditions.

\subsection{Capability Leg}\label{subsec:capability-leg}

\paragraph{Tail index $k$.} I estimate $k$ from the human completion times of the 96 tasks on which METR measures time horizons, in its Time Horizon 1.1 release, whose lengths come from timed attempts by skilled human professionals \citep{metr2026th11}. Because the horizon trend is measured on the same suite, both objects in $g_\alpha=k\,g_h$ refer to one distribution. The baseline is $k=\kbase$, with a range from $\kthin$ to $\krct$. Appendix~\ref{app:capability-detail} describes the sample, the estimator, the alternative datasets and the fits, and Figure~\ref{fig:survival} plots the fitted tail.

\paragraph{Task time-horizon growth $g_h$.} The 50\% time horizon is the task length at which a model's fitted success probability is one half \citep{kwa2025horizons}.\footnote{The cutoff success probability changes the level of the horizon but not its growth rate: \citet{kwa2025horizons} report identical doubling times for 50\% and 80\% time horizons. Since only the growth rate enters the capability ratio, it doesn't matter what cutoff we use.} On METR's updated Time Horizon 1.1 suite, the frontier time horizon has doubled every 196 days over 2019--2025 (95\% confidence interval 162 to 223 days), and every 131 days since 2023 (107 to 161 days) \citep{metr2026th11}. I use the faster post-2023 doubling times as the basis for $g_h=\ghbase$, because this better captures the window in which frontier AI systems have been applied to research tasks.

\paragraph{Software share $s$.} \citet{whitfill2025horizon} attribute a share $s=\sswbase$ of time-horizon growth to algorithmic progress, with the remainder coming from compute growth (Appendix~\ref{app:capability-detail}).

\paragraph{Substitution elasticity $\sigma$.} The model's $\sigma$ is the substitutability between research tasks, and by Proposition~\ref{prop:allocation} also between effective compute and research labor below the threshold $\sat$. Calibrated models of AI-driven growth set the task elasticity between $0.5$ and $0.67$ \citep{davidson2023compute,erdil2025gate,korinek2024scenarios}, drawing on aggregate estimates of substitution between capital and labor, and estimates of the elasticity of substitution between research compute and research labor range from Leontief to substitutes \citep{whitfillwu2025compute}. I set a baseline of $\sigma=\sigmabase$, the value in GATE, Epoch AI's integrated assessment model of AI automation \citep{erdil2025gate}, with a range from $0$, the Leontief limit, to $\sigmahigh$, the upper bound of the GATE range (Appendix~\ref{app:capability-detail}).

\subsection{Monte Carlo Simulation}\label{subsec:joint}

To quantify uncertainty in this calibration, I use Monte Carlo simulations that account for the uncertainty in each input. I draw the six inputs from the distributions in Table~\ref{tab:inputs}. The efficiency rate $g_\nu$ is drawn from the pooled log-normal \eqref{eq:pool} of Appendix~\ref{app:efficiency-detail}; its median is the baseline. The horizon rate $g_h$ is log-normal with median at the baseline and dispersion set by its confidence interval; $g_R$ and $k$ are log-uniform on their ranges, with the range for $g_R$ symmetric in logs around the baseline; $s$ and $\sigma$ are uniform on theirs. The inputs are drawn independently of one another, since they come from unrelated data sources, but the two ratios are not independent: each draw's $g_R$ is the denominator of both, so a draw with slow input growth raises $r_E$ and $r_C$ together. I take 200{,}000 draws and report, for each returns ratio, its median, its 95\% interval (the 2.5th to the 97.5th percentile) and the share of draws in which it exceeds one, together with the share of draws in each combination of legs holding. 

Using these simulations, I can measure which parameters are most important in determining whether an SIE occurs, thus guiding the focus of research and policy attention. I use a variance-based sensitivity analysis \citep{sobol2001global,saltelli2010variance}, measuring what share of the variance in explosion outcomes can be attributed to each input. Let $Y=\mathbf{1}\{r_E>1,\ r_C>1\}$ indicate that a draw of the inputs is an SIE. Its variance across draws is $p(1-p)$ where $p$ is the share of draws that are an SIE. Then the first-order Sobol index of input $x_i$ is
\[
S_i=1 - \frac{\mathbb{E}\big(\mathrm{Var}[Y\mid x_i]\big)}{\mathrm{Var}(Y)}
\]
the expected share of SIE variance removed by learning $x_i$ exactly while the other inputs remain uncertain. (The same construction applies to a group of inputs learned jointly.) I estimate the indices with a pick-freeze design \citep{saltelli2010variance}. 

\subsection{Results}\label{subsec:results}

Table~\ref{tab:headline} collects the estimates of the returns ratios and their uncertainties from the Monte Carlo simulations. Recall that an SIE requires both returns to exceed one (Corollary~\ref{cor:sie}).

\begin{table}[tbp]
\centering
\small
\begin{tabular}{@{}lcccc@{}}
\toprule
 & & \multicolumn{3}{c}{Monte Carlo} \\
\cmidrule(lr){3-5}
 & Baseline & Median & 95\% interval & Share above one \\
\midrule
Efficiency returns $r_E$ & 1.07 & 1.07 & 0.45--2.52 & 56\% \\
Capability returns $r_C$ & 3.45 & 2.01 & 0.60--8.22 & 85\% \\
Binding returns $r^{*}=\min\{r_E,r_C\}$ & 1.07 & 0.99 & 0.44--2.21 & 49\% \\
\midrule
\multicolumn{5}{@{}l}{\emph{Share of Monte Carlo draws in which}} \\
\quad both legs hold (an SIE) & & & & 49\% \\
\quad only the capability leg holds & & & & 37\% \\
\quad only the efficiency leg holds & & & & 7\% \\
\quad neither leg holds & & & & 7\% \\
\bottomrule
\end{tabular}

\caption{Estimated returns ratios and Monte Carlo simulations (Table~\ref{tab:inputs}).}
\label{tab:headline}
\end{table}

\paragraph{Efficiency leg.} At baseline, $r_E=\gnubase/\gRbase=\rEbase$, with a 95\% Monte Carlo interval of $\rElo$ to $\rEhi$. Each doubling of cumulative research input makes AI about $\rEdoubling\times$ more efficient, barely more than the $2\times$ needed to sustain the next doubling. Table~\ref{tab:published-r} places the ratio among published estimates of returns to software research, whose central estimates all lie between $0.8$ and $4$.

\begin{table}[t]
\centering
\footnotesize
\begin{threeparttable}
\begin{tabular}{@{}p{0.36\textwidth}ccp{0.27\textwidth}@{}}
\toprule
Source & Central & Range & Basis \\
\midrule
This paper & \rEbase & \rElo--\rEhi & Section~\ref{subsec:efficiency-leg} \\
\citet{eth2025software} & --- & 1--4 & domain doubling times over input growth \\
\citet{davidsonhoulden2025quick} & 1.2 & 0.4--3.6 & computer vision \\
\citet{erdil2024ideas} & 0.83 & 0.53--1.12\tnote{a} & computer chess \\
\citet{howhitfill2025experiments} & 1.89 & 1.07--3.21 & language models \\
\citet{howhitfill2025experiments} & 1.26 & 0.73--2.09 & computer vision \\
\bottomrule
\end{tabular}
\begin{tablenotes}\footnotesize
\item[a] \citet{erdil2024ideas} report only a bootstrap standard error of 0.15, from which this confidence interval is constructed via normal approximation.
\end{tablenotes}
\end{threeparttable}
\caption{$r_E$ estimates compared to published estimates on the returns to software.}
\label{tab:published-r}
\end{table}

\paragraph{Capability leg.} At baseline, $r_C=k\,s\,g_h/((1-\sigma)\,g_R)=\rCbase$, with a 95\% Monte Carlo interval of $\rClo$ to $\rChi$: the human-only task share shrinks at $\rCbase$ times the rate $(1-\sigma)g_R$ that the capability leg requires. In terms of doublings, each doubling of cumulative research input shrinks the human-only task share $1-\alpha$ by a factor of about $\rCalphadoubling$. Because tasks are complements, that raises the human-bottleneck bound on research capacity $\hbound{\alpha}$ about $\rCdoubling$-fold, far more than the $2\times$ needed to sustain the next doubling. Note that the $\rCdoubling$-fold increase is in the \textit{upper bound} on research capacity, the limit as effective compute grows without bound; it is not the magnitude of the increase in research capacity $R^*$ itself.

\paragraph{Joint uncertainty.} At baseline the efficiency leg sits just above the threshold and the capability leg is well clear of it. In the Monte Carlo simulations plotted in Figure~\ref{fig:plane}, both legs hold in \pbothpct{} of draws. This share is not the probability of an SIE; it is the share of draws from input distributions chosen to span the published evidence. Even taking the model literally, data this noisy, drawn from historical rather than current AI development, cannot support a precise probability. The capability leg alone holds in \ponlyCpct{} of draws, so that efficiency prevents an SIE; the efficiency leg alone holds in \ponlyEpct{} of draws, so that capability prevents an SIE; in \pneitherpct{} of draws, neither holds.

\begin{figure}[tbp]
\centering
\includegraphics[width=0.8\textwidth]{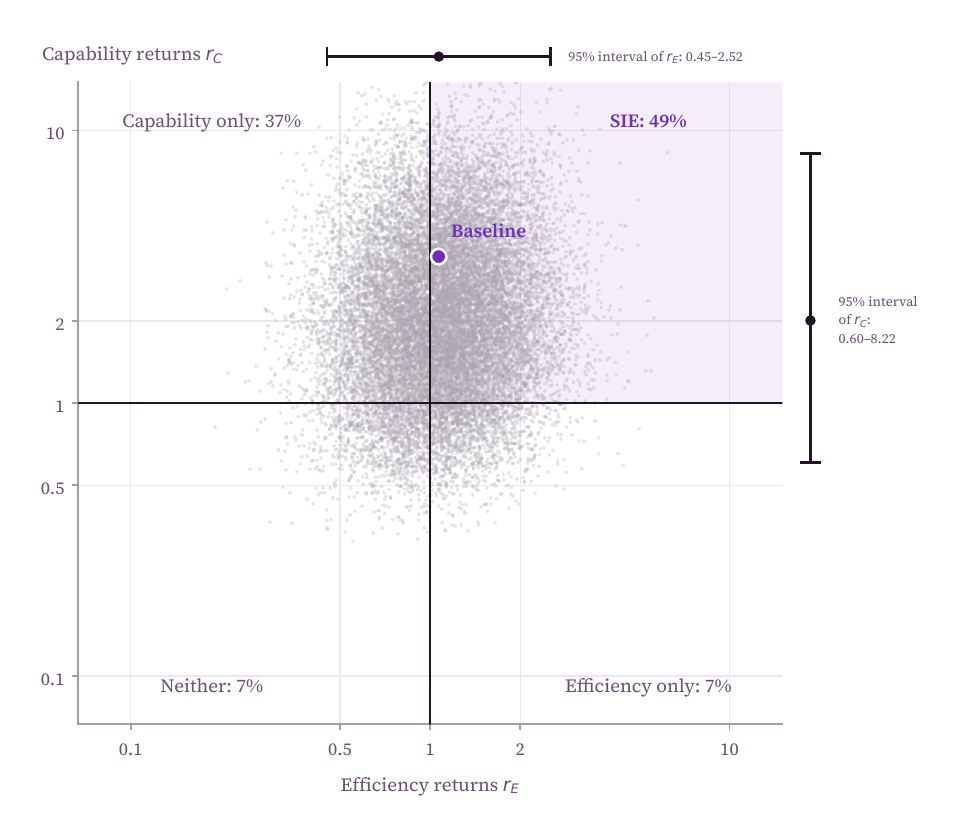}
\caption{Monte Carlo draws of the two returns ratios (Table~\ref{tab:inputs}).}
\label{fig:plane}
\end{figure}

Notably, the 95\% intervals of both returns ratios contain one: $\rElo$ to $\rEhi$ for the efficiency returns and $\rClo$ to $\rChi$ for the capability returns. This variability highlights the inherent uncertainty of my setting; even capability progress, whose point estimate seems numerically all-but-guaranteed to be sufficient for an SIE, cannot be ruled sufficient with the standard of rigor usually applied to empirical research. 

A feature of these results is that the median binding return is $r^{*}=\rstarmedian$, with a 95\% interval of $\rstarlo$ to $\rstarhi$. Thus, while both of the returns ratios are estimated to be above one, the median binding ratio is less than 1. Thus, even though each condition holds on average, they do not hold \emph{together} on average. Of course, given the amount of uncertainty in this data, there is no meaningful difference between $\rstarmedian$ and $\rEbase$. This is simply a demonstration of the importance of separating the efficiency and capability dimensions of an SIE, since focusing on only one of them would overstate the likelihood of an SIE. 

A more robust result is which leg binds. It depends only on which of the two returns is lower, not on whether either exceeds one. The efficiency returns are the lower of the two in \pEbindspct{} of draws, and at baseline the capability returns are $\rCrEratio$ times the efficiency returns. Efficiency, rather than capability, is the likely bottleneck on an SIE.

\paragraph{Sources of uncertainty.} Table~\ref{tab:sensitivity} reports the first-order Sobol index of each input---the share of variance in SIE outcome that would be eliminated by learning about that input---to measure how critical they are to resolving uncertainty about an SIE. The single largest source of uncertainty is the rate of efficiency progress $g_\nu$; learning only $g_\nu$ would eliminate \sieshareGnu{} of uncertainty about whether an SIE occurs, more than the contribution of all other inputs combined (learning all five other inputs jointly would resolve \sieshareOthers{} of uncertainty). 

Better measures of efficiency progress therefore have the highest value of information for assessing whether an SIE is plausible, and the rate of efficiency progress is also the input on which existing estimates disagree most. Thus, the single most important measurement we can do to understand the likelihood of an SIE is to measure efficiency progress in frontier AI systems. 

\begin{table}[H]
\centering
\small
\begin{tabular}{@{}lc@{}}
\toprule
Input & Share of SIE uncertainty \\
\midrule
Efficiency progress $g_\nu$ & 0.45 \\
Research input growth $g_R$ & 0.11 \\
Tail index $k$ & 0.02 \\
Substitution elasticity $\sigma$ & 0.02 \\
Software share $s$ & 0.01 \\
Time-horizon growth $g_h$ & 0.00 \\
\midrule
All inputs except $g_\nu$, jointly & 0.24 \\
\bottomrule
\end{tabular}

\caption{First-order Sobol indices of the SIE indicator across Monte Carlo draws.}
\label{tab:sensitivity}
\end{table}

\section{Conclusion}\label{sec:conclusion}

I separate two margins of AI progress: capability, the share of AI research tasks that AI can perform, and efficiency, the AI research output that a unit of compute supplies on those tasks. A software intelligence explosion requires both margins to improve fast enough relative to idea difficulty. I show that each leg of that condition is a returns ratio measurable from observable growth rates, over a historical window in which compute and AI inputs also grew, and use that to test whether the conditions for an SIE hold.

A striking feature of these results is that the likelihood of being in an SIE, \pbothpct{} of draws, is less than the likelihood of either condition holding on its own (\pEpct{} for efficiency and \pCpct{} for capability). This highlights the two-legged nature of the condition: an SIE requires both margins to improve fast enough at once, so assessing either leg in isolation overstates how likely an SIE is. My analysis should therefore be taken as evidence for the plausibility of an SIE, while highlighting its fragile nature and reliance on multiple margins of progress. 

Nonetheless, the most important takeaway from this paper is not a precise number for the likelihood of an SIE, but a focus on the most important uncertainty behind it: the rate of efficiency progress. Efficiency is the lower of the two returns ratios in \pEbindspct{} of draws, and learning only the rate of efficiency progress would resolve \sieshareGnu{} of the uncertainty about whether an SIE occurs, more than all other inputs combined. Thus, whether an SIE happens depends mostly on how fast AI is getting cheaper to run at a fixed level of performance. Better data on that rate, and on how it responds to research effort, is the most valuable step toward understanding the plausibility of an SIE.

These results also identify where policy could act to manage the possibility of an SIE. For a policymaker who wants to pace the frontier and slow AI progress, efficiency is the bottleneck target. The efficiency leg is necessary for an SIE, no matter how fast capability improves, and it is the leg that is closest to binding at baseline. A cap on efficiency progress that holds the efficiency returns at or below one, so that each doubling of cumulative research input yields at most one doubling of efficiency, rules out an SIE in the model, however fast capability improves. How much efficiency growth such a cap removes depends on how fast efficiency is in fact growing. At the baseline estimate, it requires lowering the efficiency returns by about \rEcapcut. The same logic applies to a policymaker who simply wants to control the pace of AI progress rather than prevent an SIE. The asymptotic acceleration of AI progress is governed by the binding returns $r^{*}=\min\{r_E,r_C\}$ (Proposition~\ref{prop:asymptotic}), so any reduction in the efficiency returns lowers $r^{*}$ directly, whereas a reduction in the capability returns changes it only once they fall below the efficiency returns. Thus, regulation to slow frontier AI progress could prioritize efficiency slowdowns over capability slowdowns.

This analysis is considerably more speculative than is standard in economics research, so I want to qualify these conclusions with the inherent limitations of my approach. First, this model takes a simplified view of compute. Compute enters the AI research process only as inference compute that serves AI research labor, with no distinction between the types of compute an intelligence explosion would draw on. Critically, experimental compute, one of the biggest potential bottlenecks to an SIE, appears nowhere in the model. This paper is best understood as a theory of how AI labor can contribute to an SIE, not a complete accounting of the role that compute would play in an SIE.

Second, the data used in this empirical exercise is both individually weak and also somewhat cobbled together. For example, the data on efficiency progress are noisy, with estimates of the rate of progress varying by an order of magnitude across sources. Furthermore, the relevant growth rates are computed over different periods---for example, the growth rate of research labor is measured largely since late 2022, while the inference efficiency growth rate is measured over April 2024 to November 2025, making a mismatched ratio. These are all empirical limitations inherent to the fast-moving setting of AI research, and they could influence the conclusions of my analysis. 

Third, capability in this model is the share of tasks that humans can do which AI can also do. AI that becomes better than the best human researchers at a task, for instance at choosing which research directions to pursue, counts only as an efficiency gain on that task. The model therefore leaves out an explosion driven by superhuman research taste, which features prominently in other forecasts \citep{aifutures2025takeoff}.

Finally, this paper is not about the full automation of research. The model keeps a positive human-only task share at every technology level, so it is unsuited to analyzing what would happen if AI research were fully automated. Some experts expect AI research to be close to fully automated within the next few years \citep{aifutures2025takeoff}, and a different model would be needed to analyze the effects of AI progress in such a setting.

\bibliographystyle{plainnat}
\bibliography{references}

\appendix

\section{Proofs}\label{app:proofs}

Throughout, $R^*$ denotes research capacity \eqref{eq:Rfunction} and $\sat=X/(L+X)$. The CES exponent $(\sigma-1)/\sigma$ is negative because $0<\sigma<1$. I call the two branches of \eqref{eq:Rfunction} the corner branch ($\alpha\le\sat$, where all labor works on human-only tasks) and the flat branch ($\alpha\ge\sat$). I drop time subscripts.

\subsection{Task allocation}\label{app:proof-allocation}

\begin{proof}
Since the exponent $(\sigma-1)/\sigma$ is negative, maximizing $R$ is equivalent to minimizing $\Phi \equiv \int_0^1 r(i)^{(\sigma-1)/\sigma} \dd i$, and all tasks are staffed at an optimum. Because compute contributes nothing on tasks $i > \alpha$, any optimum sets $z(i) = 0$ there, so all compute is deployed on the automated block.

\emph{Step 1 (Within-Block Equalization).} Let $\ell_A \equiv \int_0^{\alpha} \ell$ be labor assigned to automated tasks, so the automated block has total input $X + \ell_A$ and the human-only block has $L - \ell_A$. The map $u \mapsto u^{(\sigma-1)/\sigma}$ is strictly convex on $(0,\infty)$: its second derivative is $\frac{\sigma-1}{\sigma}\cdot\frac{-1}{\sigma}\, u^{(\sigma-1)/\sigma-2} > 0$. By Jensen's inequality, for a block of measure $b$ with total input $Y$,
\[
\int_{\text{block}} r^{(\sigma-1)/\sigma} \;\ge\; b \left( \frac{Y}{b} \right)^{(\sigma-1)/\sigma} = b^{1/\sigma}\, Y^{(\sigma-1)/\sigma},
\]
with equality iff $r$ is constant on the block; the constant profile is feasible because inputs are perfect substitutes within tasks. Hence
\[
\Phi(\ell_A) \;=\; \alpha^{1/\sigma} (X + \ell_A)^{(\sigma-1)/\sigma} + (1-\alpha)^{1/\sigma} (L - \ell_A)^{(\sigma-1)/\sigma}.
\]

\emph{Step 2 (Labor Across Blocks).} Write $r_A = (X+\ell_A)/\alpha$ and $r_N = (L-\ell_A)/(1-\alpha)$ for per-task inputs. Differentiating,
\[
\Phi'(\ell_A) \;=\; \frac{\sigma-1}{\sigma} \left[ \alpha^{1/\sigma}(X+\ell_A)^{-1/\sigma} - (1-\alpha)^{1/\sigma}(L-\ell_A)^{-1/\sigma} \right]
\;=\; \frac{\sigma-1}{\sigma} \left[ r_A^{-1/\sigma} - r_N^{-1/\sigma} \right].
\]
Since $(\sigma-1)/\sigma < 0$ and $u \mapsto u^{-1/\sigma}$ is strictly decreasing, $\Phi' < 0$ exactly when $r_A < r_N$: labor flows toward the block with lower per-task input. The interior solution equalizes $r_A = r_N$, i.e.\ $\ell_A = \alpha L - (1-\alpha) X$, which is nonnegative iff $\alpha \ge \sat$; then $r(i) \equiv L + X$ and $R^* = L+X$, the flat branch. If $\alpha < \sat$ then $r_A > r_N$ already at $\ell_A = 0$ (as $X/\alpha > L/(1-\alpha)$ iff $\alpha < \sat$), so $\Phi'(0) > 0$; since $\Phi(\ell_A)$ is strictly convex (a sum of strictly convex functions of affine maps), the first-order condition characterizes the unique minimum, $\Phi$ is increasing on $[0,L)$, and the corner $\ell_A = 0$ is optimal, giving $R^* = \Phi(0)^{\sigma/(\sigma-1)}$, the corner branch of \eqref{eq:Rfunction}.
\end{proof}

\subsection{Bottleneck bounds}\label{app:proof-bounds}

Define the \emph{human-bottleneck bound} as research capacity when effective compute is unlimited, $\hbound{\alpha}\equiv\lim_{X\to\infty}R^*(\alpha;L,X)$.

\begin{lemma}[Bottleneck Bounds]\label{lem:bounds}
For every $\alpha\in(0,1)$, the human-bottleneck bound exists and equals
\begin{equation}\label{eq:compute-abundant-limit}
\hbound{\alpha}=\frac{L}{(1-\alpha)^{1/(1-\sigma)}},
\end{equation}
and research capacity satisfies
\begin{equation}\label{eq:upper-bounds}
R^*(\alpha;L,X)\;\le\;\min\left\{L+X,\;\hbound{\alpha}\right\}.
\end{equation}
\end{lemma}

\begin{proof}
\emph{Compute-Abundant Limit.} Fix $\alpha\in(0,1)$. Because $\sat=X/(L+X)$ converges to one, the corner branch of \eqref{eq:Rfunction} applies for all sufficiently large $X$. Since $(\sigma-1)/\sigma<0$, $X^{(\sigma-1)/\sigma}$ converges to zero, and therefore
\[
\lim_{X\to\infty}R^*(\alpha;L,X)
=
\left[(1-\alpha)^{1/\sigma}L^{(\sigma-1)/\sigma}\right]^{\sigma/(\sigma-1)}
=
\frac{L}{(1-\alpha)^{1/(1-\sigma)}}.
\]

\emph{Upper Bounds.} $R^* \le \int_0^1 r = L + X$ is the power-mean inequality: the mean of order $(\sigma-1)/\sigma<1$ is below the arithmetic mean. For the second: for any feasible allocation, the human-only block receives only labor, total at most $L$, so by Jensen (Step 1 of Appendix~\ref{app:proof-allocation})
\[
\Phi \;\ge\; \int_{\alpha}^{1} r^{(\sigma-1)/\sigma} \;\ge\; (1-\alpha)^{1/\sigma} \left( \int_{\alpha}^1 r \right)^{(\sigma-1)/\sigma} \;\ge\; (1-\alpha)^{1/\sigma} L^{(\sigma-1)/\sigma},
\]
the last step because $u \mapsto u^{(\sigma-1)/\sigma}$ is decreasing and $\int_{\alpha}^1 r \le L$. Raising to the negative power $\sigma/(\sigma-1)$ flips the inequality: $R^* \le L (1-\alpha)^{1/(\sigma - 1)} = \hbound{\alpha}$.
\end{proof}

\subsection{Asymptotic growth}\label{app:proof-asymptotic}

\begin{proof}
Along any path, research output is research capacity at the current capability and effective compute, $R_t=R^*(\alpha(A_t);L,X(A_t))$ with $X(A)=cZA^{\eta}$ by \eqref{eq:depth} and $1-\alpha(A)=A^{-\mu}$ by \eqref{eq:breadth}, so $g_A=\delta (R^*)^{\lambda}A^{-\beta}$ and
\[
\frac{\dd\ln g_A}{\dd\ln A}=\lambda\,\frac{\dd\ln R^*}{\dd\ln A}-\beta .
\]
Let $m\equiv\min\{\eta,\mu/(1-\sigma)\}$. Since $\beta(r^*-1)=\lambda m-\beta$, it suffices to show that the elasticity of $R^*$ with respect to $A$ converges to $m$. I also show that $R^*/A^{m}$ converges to a positive constant, which the proof of Corollary~\ref{cor:sie} uses.

If $\mu=0$, then $\alpha\equiv0$ and $R^*=L$: the elasticity is zero and $m=0$.

Suppose $\mu>0$. On the corner branch, let
\[
\Phi_A=\alpha^{1/\sigma}X^{(\sigma-1)/\sigma},
\qquad
\Phi_N=(1-\alpha)^{1/\sigma}L^{(\sigma-1)/\sigma},
\qquad
\omega_A=\frac{\Phi_A}{\Phi_A+\Phi_N},
\qquad
\omega_N=1-\omega_A ,
\]
so that $R^*=(\Phi_A+\Phi_N)^{\sigma/(\sigma-1)}$. Using $\dd\ln X/\dd\ln A=\eta$ and $\dd\ln(1-\alpha)/\dd\ln A=-\mu$,
\begin{equation}\label{eq:proof-elasticity}
\frac{\dd\ln R^*}{\dd\ln A}
=
\frac{\sigma}{\sigma-1}
\left[
\omega_A\left(\frac{\sigma-1}{\sigma}\eta+\frac{1}{\sigma}\frac{\dd\ln\alpha}{\dd\ln A}\right)
-\omega_N\frac{\mu}{\sigma}
\right]
=
\omega_A\,\eta+\omega_N\frac{\mu}{1-\sigma}-\frac{\omega_A}{1-\sigma}\frac{\dd\ln\alpha}{\dd\ln A}.
\end{equation}
The last term vanishes as $A\to\infty$ because $\alpha\to1$. The weights are governed by the ratio
\[
\frac{\Phi_A}{\Phi_N}=\alpha^{1/\sigma}\left(\frac{cZ}{L}\right)^{(\sigma-1)/\sigma}A^{\left(\mu-(1-\sigma)\eta\right)/\sigma},
\]
which diverges if $\eta<\mu/(1-\sigma)$, converges to zero if $\eta>\mu/(1-\sigma)$, and converges to a positive constant if the two are equal. In the first case $\omega_A\to1$ and the elasticity converges to $\eta$; in the second $\omega_N\to1$ and it converges to $\mu/(1-\sigma)$; in the third both terms converge to $\eta$. In every case the limit is $m$. For the level: if $\omega_A\to1$, then $R^*/A^{\eta}=\alpha^{1/(\sigma-1)}cZ\,(1+\Phi_N/\Phi_A)^{\sigma/(\sigma-1)}\to cZ$, so $R^*/A^{m}\to cZ$; if $\omega_N\to1$, then $R^*\to \Phi_N^{\sigma/(\sigma-1)}=LA^{\mu/(1-\sigma)}$, so $R^*/A^{m}\to L$; in the tie, $R^*/A^{\eta}$ converges to the positive constant $\bigl[(cZ)^{(\sigma-1)/\sigma}+L^{(\sigma-1)/\sigma}\bigr]^{\sigma/(\sigma-1)}$.

On the flat branch, $R^*=L+cZA^{\eta}$, whose elasticity $\eta\,cZA^{\eta}/(L+cZA^{\eta})$ converges to $\eta$ when $\eta>0$ and equals zero when $\eta=0$, and $R^*/A^{\eta}$ converges to $cZ$ if $\eta>0$ and equals $L+cZ$ if $\eta=0$. The flat branch applies when $\alpha\ge\sat$, that is, when $1+cZA^{\eta}/L\le A^{\mu}$; for large $A$ this requires $\eta\le\mu$, which implies $\eta\le\mu/(1-\sigma)$, so on the flat branch $m=\eta$ and the limits agree with those above. The sign of $A^{\mu}-1-cZA^{\eta}/L$ is eventually constant (positive if $\eta<\mu$, negative if $\eta>\mu$, and, if $\eta=\mu$, positive if and only if $cZ<L$), so exactly one branch applies for all large $A$. On that branch the elasticity of $R^*$ converges to $m$ and $R^*/A^{m}$ converges to a positive constant, which proves the first form of \eqref{eq:asymptotic}; dividing by $\beta$ gives the second.
\end{proof}

\subsection{SIE conditions}\label{app:proof-sie}

\begin{proof}
By Proposition~\ref{prop:asymptotic}, $\liminf_{A\to\infty}\dd\ln g_A/\dd\ln A=\beta(r^*-1)$, which is positive if and only if $r^*>1$, that is, if and only if both inequalities in \eqref{eq:jointgeometric} hold. If $r^*=1$ then $\lambda m=\beta$, and since $R^*/A^{m}$ converges to a positive constant $K$ (proof of Proposition~\ref{prop:asymptotic}), $g_A=\delta (R^*/A^{m})^{\lambda}A^{\lambda m-\beta}\to\delta K^{\lambda}>0$.

For finite-time divergence, note first that $A_t\to\infty$: since $R^*>0$, $\dot A>0$, and if $A_t$ were bounded it would converge to some $A_\infty<\infty$ with $\dot A_t\to\delta R^*(A_\infty)^{\lambda}A_\infty^{1-\beta}>0$, a contradiction. So $A_t$ eventually enters any region $A\ge\bar A$, and the same convergence gives constants $0<c_0<C_0<\infty$ with $c_0A^{m}\le R^*\le C_0A^{m}$ for all sufficiently large $A$. If $\lambda m>\beta$, then $\dot A\ge \delta c_0^{\lambda} A^{1+(\lambda m-\beta)}$, so comparison with a superlinear power law gives divergence in finite time. If $\lambda m\le\beta$, then $\dot A\le\delta C_0^{\lambda} A^{1+(\lambda m-\beta)}\le\delta C_0^{\lambda} A$ for $A\ge1$, so $A_t$ grows at most exponentially and cannot diverge in finite time. Finite-time divergence therefore occurs exactly when $\lambda m>\beta$, that is, when both inequalities in \eqref{eq:jointgeometric} hold.
\end{proof}

\subsection{Capability from time horizons}\label{app:proof-horizon}

\begin{proof}
By (b) the tasks automated at level $A$ are exactly those with $q\le h(A)$, and by (c) their measure in the continuum is $\alpha(A)=F(h(A))$. For $h(A)\ge h_0$, (d) gives $1-\alpha(A)=1-F(h(A))=\bar F_0(h(A)/h_0)^{-k}$, which is \eqref{eq:alpha-horizon}. Taking logarithms, $\ln(1-\alpha)=-k\ln h+\ln\bar F_0+k\ln h_0$; differentiating with respect to $\ln A$ gives $\mu\equiv-\dd\ln(1-\alpha)/\dd\ln A=k\,\dd\ln h/\dd\ln A$, and differentiating with respect to time gives \eqref{eq:galpha}.
\end{proof}

\subsection{Returns ratios}\label{app:proof-returns}

\begin{proof}
From \eqref{eq:innovation}, $g_A=\delta R^{\lambda} A^{-\beta}$, so $\ln g_A=\ln\delta+\lambda\ln R-\beta\ln A$ and $\dd\ln g_A/\dd t=\lambda g_R-\beta g_A$. A constant $g_A$ requires $\beta g_A=\lambda g_R$.\footnote{With $L$ and $Z$ fixed, a balanced path with constant $g_A>0$ exists only in the knife-edge case $r^*=1$; the identification instead uses a historical window in which research labor grew. Over that window AI supplied a small share of research input, so $R\approx L$ and $g_R$ is the growth rate of research labor (Section~\ref{subsec:efficiency-leg}).} Dividing $g_\nu=\eta g_A$ and $g_\alpha=\mu g_A$ by $\beta g_A=\lambda g_R$ gives $\eta/\beta=g_\nu/(\lambda g_R)$ and $\mu/\beta=g_\alpha/(\lambda g_R)$, so
\[
r_E=\frac{\lambda\eta}{\beta}=\frac{g_\nu}{g_R},
\qquad
r_C=\frac{\lambda\mu}{(1-\sigma)\beta}=\frac{g_\alpha}{(1-\sigma)g_R},
\]
and Corollary~\ref{cor:sie} gives \eqref{eq:returns}.
\end{proof}

\section{Calibration Details}\label{app:calibration}

\subsection{Efficiency-leg detail}\label{app:efficiency-detail}

\paragraph{Sources.} The four estimates of efficiency progress are:
\begin{enumerate}[label=\arabic*.]
\item \emph{Inference prices.} The hardware-adjusted algorithmic component of the decline in open-weight inference prices at fixed GPQA Diamond performance is $3.22\times$ per year (90\% confidence interval $1.54$ to $6.75\times$), $\gnuinference$ per year \citep{gundlach2026price}.
\item \emph{Pretraining compute.} \citet{ho2024algorithmic} estimate, across 231 language models from 2012 to 2023, that the pretraining compute required to reach a fixed performance halved every 8.4 months (95\% confidence interval 4.5 to 14.3 months), a rate of $\gnupretraining$ per year.
\item \emph{Scale-independent component.} \citet{gundlach2025origin} retrain historical architectures with and without individual innovations and extrapolate the comparisons to frontier compute. Of the $6{,}930\times$ efficiency gain they account for over 2012--2023, $2{,}700\times$ comes from scale-dependent changes, leaving about $2.6\times$ that is independent of scale, $\gnuscalefree$ per year. Notably, this is much lower than other measures of efficiency progress.
\item \emph{Computer vision.} The compute needed to achieve a fixed performance on ImageNet classification halves every 9 months (95\% confidence interval 4 to 25 months), giving a progress rate of $\gnuimagenet$ per year \citep{erdil2022cv}.
\end{enumerate}

No single one of these measures the model's $g_\nu$ exactly. While inference efficiency is the ideal measure, \citet{gundlach2026price} is not an ideal estimate of it, focusing only on two benchmarks and identifying the inference decline only from a sliver of open models over a short time period. In contrast, the other three measures cover roughly a decade of progress from 2012 onward, measure compute directly in FLOPs rather than inferring it from prices, and hold performance fixed by design, making them suitable for triangulation. Thus, I assign $1/2$ weight to inference price declines, and $1/6$ weight to the three other measures, pooling them linearly as described below. The baseline value that comes from averaging the rates this way is $g_\nu=\gnubase$.

\paragraph{Pooling the efficiency estimates.} Source $j$ reports a rate $\hat g_j$ with log standard deviation $s_j$ and receives weight $w_j$. The baseline is the weighted average of the rates, and $g_\nu$ is log-normal around it with a dispersion equal to the weighted average of the sources' dispersions:
\begin{equation}\label{eq:pool}
\hat g_\nu=\sum_j w_j\,\hat g_j,
\qquad
\ln g_\nu\;\sim\;\mathcal{N}\Big(\ln\hat g_\nu,\ \big(\textstyle\sum_j w_j s_j\big)^2\Big).
\end{equation}
Each $s_j$ comes from its source's confidence interval; the scale-independent component, for which none is reported, carries the dispersion of the pretraining interval. I pool linearly rather than log-linearly because a log-linear pool weights each source by its precision as well as by $w_j$, and the precision of the scale-independent component is borrowed from the pretraining interval rather than reported; a linear pool keeps the weights at those of Section~\ref{subsec:efficiency-leg}.

\paragraph{Computer-vision and inference-price estimates.} The fixed-performance design applied to ImageNet classification gives a halving time of 9 months with a 95\% confidence interval of 4 to 25 months \citep{erdil2022cv}. On the inference side, \citet{gundlach2026price} decompose the decline in the price of reaching a fixed score on GPQA Diamond (graduate-level science questions), Mock AIME (competition mathematics) and SWE-bench Verified (real software-engineering issues from GitHub repositories) into a hardware component, about $1.3\times$ per year in price-performance, and an algorithmic component, about $3\times$ per year for open-weight models ($3.22\times$ on GPQA Diamond, the rate used here), with strong dependence on the score level. Raw declines in the price of fixed capability are far larger, $9\times$ to $900\times$ per year across benchmarks and thresholds \citep{cottier2025inferenceprices}, because they also reflect hardware, pricing, and the conversion of training compute into inference savings through distillation and overtraining, which \citet{villalobos2023trading} put at roughly one order of magnitude of inference compute per one to two orders of magnitude of additional training compute. The hardware-adjusted algorithmic component is close to the pretraining rate, so the training-side and inference-side measurements of fixed-task efficiency agree, and both enter the pool that sets the baseline.

\paragraph{Denominator alternatives.} \citet{howhitfill2025experiments} report that AI companies are almost doubling their staff headcounts each year, using OpenAI's staff counts since November 2022, and \citet{davidson2026networks} report growth in AI-laboratory headcount of roughly $140\%$ per year since ChatGPT, from Epoch AI data. The software-researcher growth rate of $110\%$ per year in \citet{davidson2026networks} is backed out from their calibration rather than measured. Research and engineering staff across the sector grow at about $23\%$ per year. A denominator that treats experiment compute as a research input, with a compute share of research spending of about two thirds and compute spending growing at $3.7\times$ per year, is $0.67\times1.3+0.33\times0.85\approx1.15$ \citep{howhitfill2025experiments}. Experiment compute is not the inference compute $Z$ of the model, so the labor rate is the baseline and the augmented rate is the upper end of the range in Table~\ref{tab:inputs}; the lower end, $\gRlow$ ($\gRlowfactor\times$ per year), mirrors it in logs around the baseline.

\subsection{Capability-leg detail}\label{app:capability-detail}

\paragraph{Task-length sample.} The baseline source is the task suite on which METR measures time horizons, in its Time Horizon 1.1 release \citep{metr2026th11}: 157 software, machine-learning, cybersecurity and reasoning tasks from HCAST \citep{rein2025hcast}, five research-engineering environments from RE-Bench \citep{wijk2024rebench}, and 66 sub-minute software actions. A task's length is the geometric mean of the completion times of the human baseliners who succeeded at it. Tasks with no successful baseline carry only a researcher estimate, and I drop them: the estimation sample is the 96 HCAST and RE-Bench tasks whose lengths are measured rather than estimated. This suite is the baseline for two reasons. First, among the datasets with per-task times its measured lengths reach furthest into the tail, where the index is fitted, and they are measured under one protocol. Second, the horizon trend is measured on it, so both objects in $g_\alpha=k\,g_h$ refer to one distribution. Appendix~\ref{app:data} describes the baselining protocol and the two other datasets.

Two further datasets serve to triangulate:
\begin{enumerate}[label=\arabic*)]
\item METR's early-2025 study of experienced open-source developers \citep{becker2025productivity} records forecast and actual completion times for 246 issues from the developers' own repositories: real work rather than a designed suite, but developers were asked to contribute issues expected to take at most two hours, so the distribution is truncated from above.
\item GDPval \citep{openai2025gdpval} covers 1,320 tasks across 44 occupations, with completion times self-reported by experts and validated by reviewers. Its tasks are far from AI research, but it spans a very diverse set of professions, and so it represents an outside view of how long tasks are distributed; only its published quantiles are available.
\end{enumerate}

\begin{table}[t]
\centering
\footnotesize
\setlength{\tabcolsep}{4pt}\begin{tabular}{@{}lcccccccc@{}}
\toprule
 & & \multicolumn{5}{c}{task length, hours} & \multicolumn{2}{c}{share of hours in tasks} \\
\cmidrule(lr){3-7}\cmidrule(lr){8-9}
Data & $n$ & p25 & p50 & p75 & p90 & max & $>4$ h & $>8$ h \\
\midrule
TH1.1 baselined tasks & 96 & 0.1 & 0.4 & 1.3 & 5.5 & 18.5 & 0.74 & 0.37 \\
Early-2025 RCT, actual & 246 & 0.7 & 1.2 & 2.2 & 3.8 & 10.0 & 0.32 & 0.10 \\
Early-2025 RCT, forecast & 246 & 1.0 & 1.5 & 2.5 & 4.0 & 6.0 & 0.09 & 0.00 \\
GDPval, full set (published) & 1320 & 2.0 & 4.0 & 8.0 & --- & 605.0 & --- & --- \\
\bottomrule
\end{tabular}

\caption{Task-length data (human hours).}
\label{tab:task-summary}
\end{table}

\paragraph{Estimator.} Let $x_1,\dots,x_n$ be the human completion times of the tasks in a dataset. The empirical task-length distribution, weighted by task length as in step (c), satisfies
\begin{equation}\label{eq:survival}
1-\hat F(T)=\frac{\sum_{i:\,x_i>T}x_i}{\sum_i x_i},
\end{equation}
the share of work-hours in tasks longer than $T$. If $F$ has the Pareto tail of step (d) above $h_0$, the length-weighted (pseudo-)maximum-likelihood estimate of $k$ from the tasks above $h_0$, which is consistent for $k$, is
\begin{equation}\label{eq:hill-time}
\hat k=\left[\frac{\sum_{i:\,x_i>h_0}x_i\ln(x_i/h_0)}{\sum_{i:\,x_i>h_0}x_i}\right]^{-1}.
\end{equation}
As a check on the functional form, I also report a log-log regression of $1-\hat F$ on $T$ between $h_0$ and the longest task. The threshold $h_0$ is four hours: above the body of the suite, where three quarters of the baselined tasks are shorter than $1.3$ hours, and below the eight-hour budget that caps the RE-Bench tasks.

\paragraph{Fits.} Table~\ref{tab:tail-fits} reports the fits. The estimate above four hours is $k=\kbase$, and the log-log slope gives $\kbaseloglog$. The developer study gives $\krct$ above four hours; its distribution is narrower because issues were screened to be completable in a session. GDPval's quantiles give $\kgdpval$ between the 75th percentile and the maximum; its published quantiles support only a count-weighted fit, so the value is not directly comparable to the length-weighted baseline. Thus, the baseline is $k=\kbase$, with a range from the same estimator at a one-hour threshold, $\kthin$, to the developer-study value, $\krct$.

\begin{figure}[tbp]
\centering
\includegraphics[width=\textwidth]{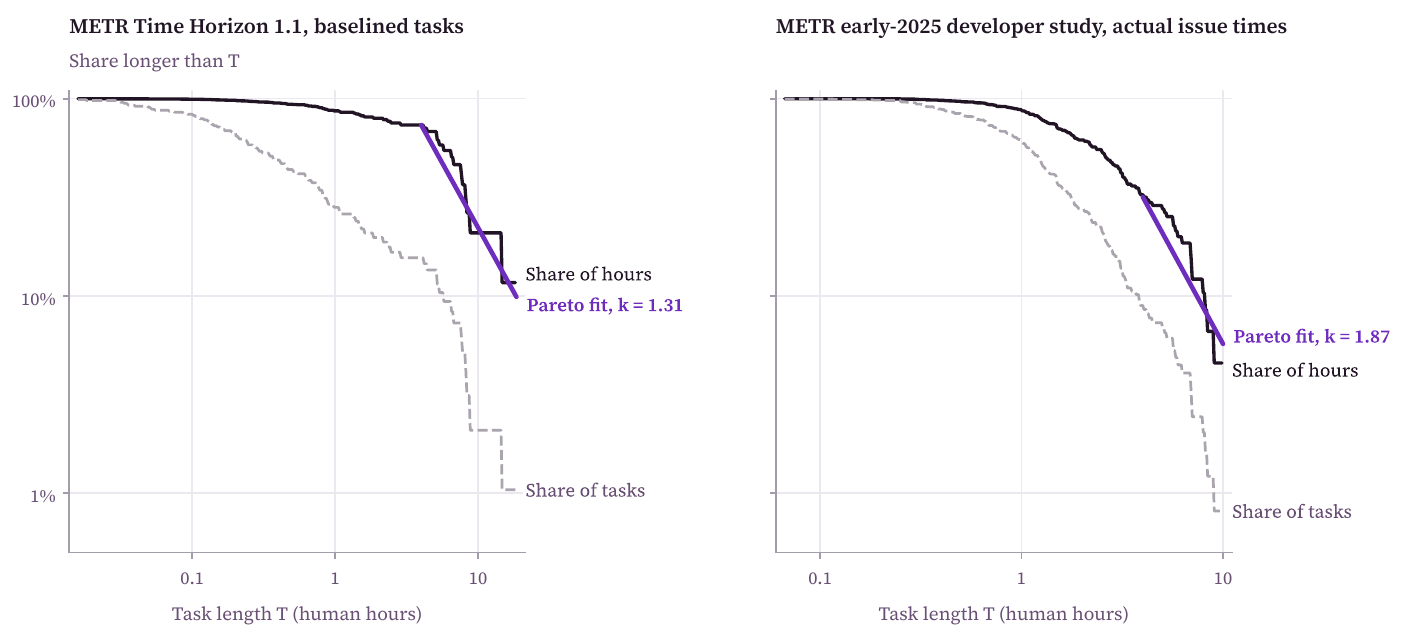}
\caption{Survival of task length by share of hours (solid) and share of tasks (dashed), with the Pareto tail fitted above four hours (purple).}
\label{fig:survival}
\end{figure}

\paragraph{Software share.} In \citet{whitfill2025horizon}, algorithmic progress of $g_{\mathrm{alg}}=1.27$ per year (95\% confidence interval $0.67$ to $3.61$), against frontier training-compute growth of $4.5\times$ per year \citep{sevilla2024compute}, gives $s=1.27/(1.27+\ln 4.5)=\sswbase$, and their confidence interval maps to a range of $\sswlow$ to $\sswhigh$. Two independent decompositions serve for triangulation: the stitched capability index of \citet{ho2025rosetta} attributes $0.30$ units per year of capability growth to software and $0.27$ to compute, a share of $0.52$, and the pretraining efficiency rate of \citet{ho2024algorithmic} in place of $g_{\mathrm{alg}}$ gives a share of $0.40$.

\paragraph{Substitution elasticity.} The model's $\sigma$ is the substitutability of AI research labor, supplied by inference compute, for human researchers across research tasks. No estimate of $\sigma$ exists. Calibrated models of AI-driven growth set the task elasticity at $0.67$ \citep{davidson2023compute}, $0.61$ with a range of $0.17$ to $0.83$ \citep{erdil2025gate}, and $0.5$ \citep{korinek2024scenarios}, drawing on aggregate estimates of substitution between capital and labor. \citet{whitfillwu2025compute} estimate the elasticity of substitution between research compute and research labor: $2.58$ (standard error $0.34$) treating all research compute as the input, and $-0.10$ (standard error $0.18$), which is Leontief, treating frontier-scale experiments as the input. Their compute is the compute spent on training runs and experiments, which is not the inference compute $Z$ of the model, so these estimates bear on $\sigma$ indirectly, through the complementarity between machine and human inputs in research production. By Proposition~\ref{prop:allocation}, $\sigma$ is also the elasticity of substitution between effective compute and research labor below the threshold $\sat$. The baseline is the GATE value $\sigma=\sigmabase$, with a range from $0$, the Leontief limit (in which the results above hold with $\sigma=0$) at the low end of the \citet{whitfillwu2025compute} evidence, to $\sigmahigh$, the upper bound of the GATE range. At $\sigma\ge1$, outside the model's maintained range, tasks are substitutes rather than complements and the capability leg holds trivially.

\subsection{Task-length data}\label{app:data}

\paragraph{METR task suite.} HCAST \citep{rein2025hcast} comprises 189 tasks in 78 families spanning machine learning, cybersecurity, software engineering and general reasoning; the time-horizon suite of \citet{kwa2025horizons} used 97 of them, and Time Horizon 1.1 \citep{metr2026th11} raised the total suite from 170 to 228 tasks by adding 73 HCAST tasks, removing 15 and updating 53. Human baseliners are recruited mainly through the professional networks of METR staff, must hold a degree from a top-100 university or have more than three years of relevant professional experience, and must pass a domain-specific qualification task; their average relevant experience is about five years. They are paid \$50 to \$100 per hour plus \$25 to \$150 per hour in performance bonuses, with a success bonus proportional to time spent and a speed bonus for completing a task faster than every other baseliner, a structure METR describes as pushing task-time estimates toward a lower bound on human time. Baseliners work in the same containerized environment as the agents over SSH with identical instructions, are recorded, may use the web, and may not use AI assistants; RE-Bench baseliners were allowed AI tools. Most tasks are scored pass or fail; continuously scored tasks are binarized at a success threshold shown to the baseliners. A task's length is the geometric mean of the completion times of its successful baseline runs, since baseline times are roughly log-normal; a task counts as baselined if at least one run succeeded, and METR notes that conditioning on success biases lengths downward. For tasks with no successful baseline the length is a forecast by the task's creator or an author familiar with it. In HCAST, 61\% of baseline attempts succeeded and 110 of 189 tasks received at least one successful baseline. Tasks are weighted by the inverse square root of their family size in the horizon regression. METR rates each task on sixteen ``messiness'' factors; the mean score is 3.2 of 16, no task scores above 8, and METR states that the tasks differ systematically from real work in having automatic scoring, clearly defined scope and little outside context.

Three features of this protocol matter for the tail fit. First, RE-Bench environments are run under an eight-hour budget, and the long HCAST additions are mostly estimated rather than baselined: of the 31 tasks in Time Horizon 1.1 estimated at eight hours or more, five have measured lengths. The estimation sample therefore contains no measured length above 18.5 hours, and the tail is censored there. Second, the geometric mean over successful runs and the speed incentives both shorten measured lengths relative to the time a typical human would need, which shifts the distribution left without a clear effect on its slope. Third, I exclude the sub-minute software actions because they were baselined by METR employees under a different protocol and lie four orders of magnitude below the threshold. The estimation sample is the 96 tasks in the Time Horizon 1.1 release whose length comes from human baselines and whose source is HCAST or RE-Bench.

\paragraph{Developer study.} \citet{becker2025productivity} recruited 16 experienced developers, typically with more than a decade of experience and about five years and 1,500 commits on their repository, from mature open-source projects averaging 23,000 stars, 1.1 million lines of code and ten years of age. Each developer supplied a list of real issues from their own repository; about half were not formally defined on the repository. Before each issue was randomized, the developer forecast how long it would take with and without AI assistance. Issues were then assigned by coin flip to an AI-allowed or AI-disallowed condition, and developers self-reported the time spent before and after pull-request review; screen recordings were collected for compliance and for activity labeling of a subset. Developers were paid \$150 per hour and were asked to contribute issues taking at most two hours or to split longer ones, so the sample is drawn from the short end of the work distribution, and the longest recorded implementation took ten hours. I use the initial implementation time across both conditions. METR's late-2025 follow-on with 57 developers is larger but, by METR's own account, subject to selective submission of tasks and unreliable time reporting, so I do not use it.

\paragraph{GDPval.} \citet{openai2025gdpval} constructed 1,320 tasks across 44 occupations in the nine sectors that contribute most to United States GDP, with at least 30 tasks per occupation, and a gold subset of 220. Task writers were professionals with at least four years of experience, fourteen on average, screened by interview, background check and a qualification quiz. Completion times are the writers' self-reports of the real-world time the task takes, validated and corrected by multiple independent occupational reviewers. The paper publishes the mean, standard deviation, minimum, quartiles and maximum of these times for the full set and the gold subset; the public release contains prompts, reference files and rubrics for the gold subset but no times. I fit the full-set quartiles and maximum, count-weighted by necessity.

\begin{table}[h]
\centering
\small
\begin{tabular}{@{}lcccccc@{}}
\toprule
 & & & \multicolumn{1}{c}{count-weighted} & \multicolumn{2}{c}{length-weighted} \\
\cmidrule(lr){4-4}\cmidrule(lr){5-6}
Data & Threshold & Tasks above & Hill & Hill & log-log \\
\midrule
TH1.1 baselined tasks ($n = 96$) & 1 h & 27 & 0.76 & 0.52 & 0.67 \\
 & 2 h & 19 & 0.96 & 0.73 & 0.97 \\
 & 4 h & 15 & 1.74 & 1.31 & 1.29 \\
Early-2025 RCT issues ($n = 246$) & 1 h & 148 & 1.33 & 0.86 & 1.15 \\
 & 2 h & 68 & 1.72 & 1.25 & 1.51 \\
 & 4 h & 21 & 2.13 & 1.87 & 2.29 \\
GDPval quantiles ($n = 1320$) & 8 h & 330 & 1.34 & --- & --- \\
\bottomrule
\end{tabular}

\caption{Tail index estimates for task length: count-weighted and length-weighted Hill estimates, and the log-log slope of $1-\hat F$ between the threshold and the longest task.}
\label{tab:tail-fits}
\end{table}

\subsection{Monte Carlo}\label{app:mc}

Table~\ref{tab:sources} gives the source of every input in Table~\ref{tab:inputs}. The Monte Carlo draws $g_\nu$ from the linear pool \eqref{eq:pool} of the rates $\gnuinference$, $\gnuscalefree$, $\gnupretraining$ and $\gnuimagenet$, with weights $1/2$ on the first and $1/6$ on each of the others: a log-normal distribution with median $\gnubase$, the baseline, and log standard deviation $\gnusd$, the weighted average of the sources' log standard deviations, each implied by its source's confidence interval, with the scale-independent component, for which none is reported, taking that of \citet{ho2024algorithmic}; $g_h$ from a log-normal distribution with median at the baseline and dispersion set by its confidence interval; $g_R$ and $k$ from log-uniform distributions on their ranges; and $s$ and $\sigma$ from uniform distributions on their ranges. Inputs are drawn independently of one another, and each draw's $g_R$ enters both ratios.

\begin{table}[h]
\centering
\footnotesize
\begin{tabular}{@{}lp{0.86\textwidth}@{}}
\toprule
Input & Source \\
\midrule
$g_\nu$ & Linear pool of four fixed-task rates: the baseline is their weighted average and the log standard deviation is the weighted average of their log standard deviations, each from its source interval. Sources: the hardware-adjusted algorithmic component of inference-price declines of \citet{gundlach2026price}, $3.22\times$/yr on GPQA Diamond for open-weight models (90\% CI $1.54$--$6.75$) (weight 1/2); the scale-independent component of \citet{gundlach2025origin}, $2.6\times$ over 2012--2023, no interval reported, pretraining dispersion (1/6); \citet{ho2024algorithmic}, compute to fixed perplexity halves every 8.4 months (95\% CI 4.5--14.3) (1/6); the ImageNet rate of \citet{erdil2022cv}, 9 months (95\% CI 4--25) (1/6) \\
$g_R$ & Frontier-lab staff growth since late 2022: OpenAI staff $0.85$/yr \citep{howhitfill2025experiments}, Epoch AI lab headcount $\approx140\%$/yr \citep{davidson2026networks}; baseline is their mean in logs; high end is the experiment-compute-augmented denominator of \citet{howhitfill2025experiments}, low end mirrors it in logs around the baseline \\
$g_h$ & \citet{metr2026th11}: doubling every 131 days since 2023 (95\% CI 107--161); the 2019--2025 doubling time is 196 days (162--223) \\
$s$ & \citet{whitfill2025horizon}: horizon growth proportional to effective-compute growth $g_C + g_{\mathrm{alg}}$, with algorithmic progress $g_{\mathrm{alg}} = 1.27$/yr (95\% CI 0.67--3.61) and training compute growth $4.5\times$/yr \citep{sevilla2024compute}; the capability-index decomposition of \citet{ho2025rosetta} gives $0.52$ \\
$k$ & Length-weighted Pareto fit above four hours to the human-baselined tasks of METR's Time Horizon 1.1 suite \citep{metr2026th11,rein2025hcast}; low end: the same fit above one hour; high end: the same fit to the real open-source issues of \citet{becker2025productivity}; GDPval quantiles \citep{openai2025gdpval} give a count-weighted $1.34$ \\
$\sigma$ & \citet{whitfillwu2025compute}: $-0.10$ (SE $0.18$) with frontier-experiment compute, $2.58$ (SE $0.34$) with total research compute; \citet{davidson2023compute}: $0.67$; \citet{erdil2025gate}: $0.61$ ($0.17$--$0.83$); \citet{korinek2024scenarios}: $0.5$ \\
\bottomrule
\end{tabular}

\caption{Sources for the calibration inputs.}
\label{tab:sources}
\end{table}

\end{document}